\documentclass[a4paper,aps,amssymb,amsmath,amsfonts,superscriptaddress,
twocolumn]{revtex4-2}
\pdfoutput=1

\usepackage{physics}
\usepackage{graphicx}
\usepackage{bm,bbm}
\usepackage{epstopdf}
\usepackage{amsthm}
\usepackage{amsmath}
\usepackage{color}
\usepackage{amsopn}
\usepackage{amssymb}
\usepackage{hyperref}
\usepackage{subcaption}
\usepackage{enumerate}
\usepackage{todonotes}
\usepackage{ dsfont }
\usepackage[T1]{fontenc}
\usepackage[english]{babel}
\usepackage[utf8]{inputenc}
\usepackage{mathtools}
\usepackage{multirow}
\usepackage{adjustbox}
\usepackage[utf8]{inputenc}

\newtheorem{remark}{Remark}
\newtheorem{theorem}{Theorem}

\newtheorem{proposition}{Proposition}

\newcommand{\proj}[1]{\ensuremath{\ketbra{#1}{#1}}}

\newcommand{\Id}{{\rm 1\hspace{-0.9mm}l}}
\newcommand{\LL}{\mathcal{L}}

\newcommand{\RR}{\mathcal{R}}
\newcommand{\A}{\mathcal{A}}
\newcommand{\B}{\mathcal{B}}
\newcommand{\PP}{\mathcal{P}}

\newcommand{\QQ}{\mathcal{Q}}

\newcommand{\TT}{\mathcal{T}}
\renewcommand{\SS}{\mathcal{S}}

\newcommand{\XX}{\mathcal{X}}
\newcommand{\YY}{\mathcal{Y}}
\newcommand{\ZZ}{\mathcal{Z}}

\newcommand{\ketv}[1]{\ensuremath{\left|\left.{#1} \right\rangle \! 
\right\rangle}}
\newcommand{\brav}[1]{\ensuremath{\left\langle \! \left\langle {#1} 
\right.\right|}}
\newcommand{\projv}[1]{\ketv{{#1}} \!\! \brav{{#1}}}

\begin{document}
	
\title{Storage and retrieval of von Neumann measurements via indefinite 
causal order structures}
\author{Paulina Lewandowska}
\affiliation{IT4Innovations, VSB~-~Technical University of Ostrava, 17.~listopadu 2172/15, 708 33 Ostrava, Czech Republic}
\author{Ryszard Kukulski$^*$}
\affiliation{Faculty of Physics, Astronomy and Applied Computer Science,\\
ul. Łojasiewicza 11, Jagiellonian University, 30-348 Kraków, Poland}
\email{ryszard.kukulski@uj.edu.pl}
\begin{abstract}
	This work presents the problem of learning an unknown  von 
		Neumann measurement of dimension $d$ using indefinite causal 
		structures. In the considered scenario, we have access to $N$ copies 
		of the measurement. We use formalism of process 
		matrices to store information about the 
		given measurement, that later will be used to reproduce its best possible 
		approximation. Our goal is to compute the maximum 
	value of the average fidelity function $F_d(N)$ of our procedure. We 
	prove that $F_d(N) = 1 - \Theta \left( \frac{1}{N^2}\right)$  for 
	arbitrary but fixed 
	dimension $d$. Furthermore, we present the SDP program for computing 
	$F_d(N)$. Basing on the numerical investigation, 
	we show that for the qubit von Neumann measurements using indefinite 
	causal 
	learning structures provide better approximation than quantum 
	networks, starting from $N \ge 3$.
\end{abstract}
\maketitle
\section{Introduction}
In the  setup of storage and retrieval (SAR) of quantum measurements,  we would like to
to approximate a given, unknown measurement, which we were able to perform 
$N$ times experimentally.
Such strategy usually consists of the notion of a quantum networks known 
also as quantum combs \cite{bisio2010optimal}. 
Then, the scheme is created by preparing some initial quantum state, 
applying the unknown operation $N$ times, which stores information about the 
unknown operation for later use, and finally, a retrieval operation that
produces an approximation of the black box on some arbitrary quantum state. 
The scheme is optimal when it achieves the highest possible fidelity of the 
approximation~\cite{lewandowska2022storage, raginsky2001fidelity, 
belavkin2005operational}.

The seminal work in this field was the paper~\cite{bisio2011quantum}. The Authors have shown
that, in general, the optimal algorithm for quantum measurement learning cannot 
be parallel. In \cite{lewandowska2022storage}, whereas,  the Authors have discovered the asymptotic behaviour of the maximum value of the average fidelity function over all 
possible learning schemes is given by $1 - 
\Theta\left(\frac{1}{N^2}\right)$.

In this work,  we introduce a new aspect of von Neumann measurement learning
by using indefinite causal structure theory.
The topic of indefinite causal structures has recently gained traction in 
quantum information research. This more general model of computation has the 
potential to outperform algorithms based on quantum networks in specific 
tasks, such as learning or discriminating between two quantum channels 
\cite{bavaresco2021strict, quintino2022deterministic, bavaresco2022unitary}. 
In the problem of von Neumann measurements learning, indefinite causal 
structures find a place in the storage part of the procedure. Their 
mathematical description is formalized in the language of process 
matrices~\cite{lewandowska2023strategies, araujo2015witnessing}.

As for the results of our work, we will prove that for $2 \rightarrow 1$ 
learning scheme, using indefinite causal structures does not 
improve the average fidelity function $F_d(2)$ for any dimension $d$. 
Next, however, we will show the numerical advantage of using causal 
structure theory in the $N \rightarrow 1$ learning scheme for $N \ge 3$. Finally, 
we determine the asymptotic behavior of $F_d(N)$ for $ N \rightarrow \infty$. 

This paper is organized as follows. In Section~\ref{sec:formulation} we 
formulate $N \rightarrow 1$ learning scheme of  von
Neumann measurement and 
 necessary mathematical tools needed to describe this problem.
Section~\ref{sec:cases} presents a theoretical results of this paper. In  
Subsection~\ref{sec:2-to-1}, we show that the indefinite causal structures do not 
improve the average fidelity function of learning for two 
copies of von Neumann measurement, whereas  Subsection~\ref{sec:sdp} 
presents the semidefinite programs for computing the maximum value of the 
fidelity function for a finite number of copies. In 
Subsection~\ref{sec:n-to-1},  we also prove the asymptotic behavior of the 
fidelity function. 
Section~\ref{sec:numerical} shows a numerical advantage of using the indefinite causal structure of von Neumann measurements learning over any strategies based on quantum combs. Finally, the concluding remarks are presented in  Section~\ref{sec:conclusion}. 

\section{Problem formulation}\label{sec:formulation}

This section presents the formulation of  learning scheme of an 
unknown von Neumann
measurement via indefinite causal structures and necessary mathematical tools 
needed to describe this problem. 

\subsection{Mathematical framework}\label{sec:notation}

Let us introduce the following notation. Consider two complex Euclidean spaces and
denote them by $\XX, \YY$. By $\mathrm{L}(\XX, \YY)$ we denote the collection of all linear mappings of the form $M: \XX \rightarrow \YY$. As a shorthand we put $\mathrm{L}(\XX) \coloneqq 
\mathrm{L}(\XX, \XX)$. By $\mathrm{Herm}(\XX)$ we
denote the set of Hermitian operators while the subset of $\mathrm{Herm(\XX)}$ consisting of positive
semidefinite operators will be denoted by $\mathrm{Pos}(\XX)$. The
set of quantum states defined on space $\XX$, that is the set of positive 
semidefinite operators having unit trace,
will be denoted by $\Omega(\XX)$. We will also need a linear mapping 
transforming
$\mathrm{L}(\XX) $ into $\mathrm{L}(\YY)$ as $\mathcal{T}: 
\mathrm{L}(\XX)
\mapsto \mathrm{L}(\YY).$ There exists a bijection between introduced 
linear mappings
$\mathcal{T}$ and set of matrices $\mathrm{L}(\YY \otimes \XX)$, known as the 
Choi-Jamio{\l}kowski
isomorphism~\cite{choi1975completely, jamiolkowski1972linear}. Its explicit 
form is ${T} =
\sum_{i,j} \TT(\ketbra{i}{j}) \otimes \ketbra{i}{j} $. We will denote 
linear mappings with calligraphic font $\LL, \SS, \mathcal{T}$ etc., whereas 
the 
corresponding Choi-Jamio{\l}kowski matrices as plain symbols: $L, S, T$ etc. 

A general quantum measurement (POVM) $\mathcal{Q}$ can be viewed as a set of 
positive semidefinite
operators $ \QQ = \{ Q_i \}_i$ such that $\sum_i Q_i = \Id$. These operators 
are usually called 
effects. The von Neumman measurements, $\PP_U$,
are a special subclass of measurements whose all effects are rank-one 
projections given by $\PP_U =
\{P_{U,i}\}_{i}^{} = \{U\ketbra{i}{i}U^\dagger\}_{i}^{}$ for some 
unitary matrix $U \in \mathrm{L}(\XX)$.
%
 The
Choi matrix of $\PP_U$ is $ P_U = \sum_{i} \ketbra{i}{i} \otimes 
\overline{P_{U,i}}, $ which will be
utilized throughout this work.

 Let us consider a composition of mappings
 $\RR = \mathcal{N} \circ \mathcal{M}$, where $\mathcal{N}: \mathrm{L}(\ZZ ) \rightarrow \mathrm{L}(\YY)$ and $\mathcal{M}: \mathrm{L}(\XX) \rightarrow \mathrm{L}(\ZZ)$ with Choi matrices $N \in \mathrm{L}(\ZZ \otimes \YY)$ and $M \in \mathrm{L}(\XX \otimes \ZZ)$, respectively. Then, the Choi matrix of $\RR$ is given
 by \cite{chiribella2009theoretical} $
 R = \tr_{\ZZ} \left[ \left(\Id_\YY \otimes M^{T_\ZZ} \right) \left( N \otimes \Id_\XX \right) \right], 
 $ 
 where $M^{T_\ZZ}$ denotes the partial transposition of $M$ on the subspace $\ZZ$. The above result
 can be expressed by introducing the notation of the link product of the operators $N$ and
 $M$ as \begin{equation}
 N * M \coloneqq \tr_{\ZZ} \left[\left(\Id_\YY \otimes M^{T_\ZZ} \right) \left( N \otimes \Id_\XX \right) \right].
 \end{equation}

 Finally, we define the operator  $\prescript{}{\XX}{Y}$ as 
$
\prescript{}{\XX}{Y}  = \frac{\Id_\XX}{\dim(\XX)} \otimes \tr_\XX Y,
$
 for every $Y \in \mathrm{L}(\XX \otimes \ZZ)$, where $\ZZ$ is an arbitrary complex Euclidean space and the projective operator \begin{equation}\label{proj-n}
L_V(W) =  \prescript{}{\left[1- \prod_{i} \left(1 - \A_O^i + \A_I^i \A_O^i\right) + \prod_{i} \A_I^i \A_O^i  \right]}{W}.
\end{equation}
	We say that $W \in \mathrm{Pos}(\A_I^1\otimes \A_O^1 \otimes  \ldots 
	\otimes \A_I^N \otimes \A_O^N) $ is  $N$-partite process matrix if it 
	fulfills the following conditions~\cite{araujo2015witnessing}
	\begin{equation}
	 W  = L_V(W) \,\,\, \text{and}\,\,\tr(W) = \dim(\A_O^1) \cdot \ldots \cdot \dim(\A_O^N),
	\end{equation}
	where the projection operator $L_V$ 
	is defined by Eq.~\eqref{proj-n}.

\subsection{Learning setup}\label{sec:setup}
Imagine we are given a black box with the promise that it contains some von 
Neumann measurement,
$\PP_U$, which is parameterized by a unitary matrix $U$. The exact value of $U$ is unknown to us. We
are allowed to use the black box $N$ times. Our goal is to prepare a storage strategy $\SS $ and 
a retrieval measurement $\RR$
such that we are able to approximate $\PP_U$ on an arbitrary state $\rho \in \Omega(\XX_{in})$.  This 
approximation will
be denoted throughout this work as $\QQ_U$. We would like to point out that, 
generally, $\QQ_U$ will
\emph{not} be a von Neumann measurement. The learning scheme will be denoted by $\LL$ with Choi matrix $L$ being a concatenation $L = R * S$. Additionally, we assume that the Choi matrix of the storage  $S$ has an indefinite causal order. More precisely, 
the storage $S$ 
is described by  $N-$partite process matrix $W$, that is $W = \tr_{\XX_a} S$.
We provide an overview of the learning scheme in 
Fig.~\ref{fig:schema}.

\begin{figure}[!ht]
	\centering
	\includegraphics[scale=0.5]{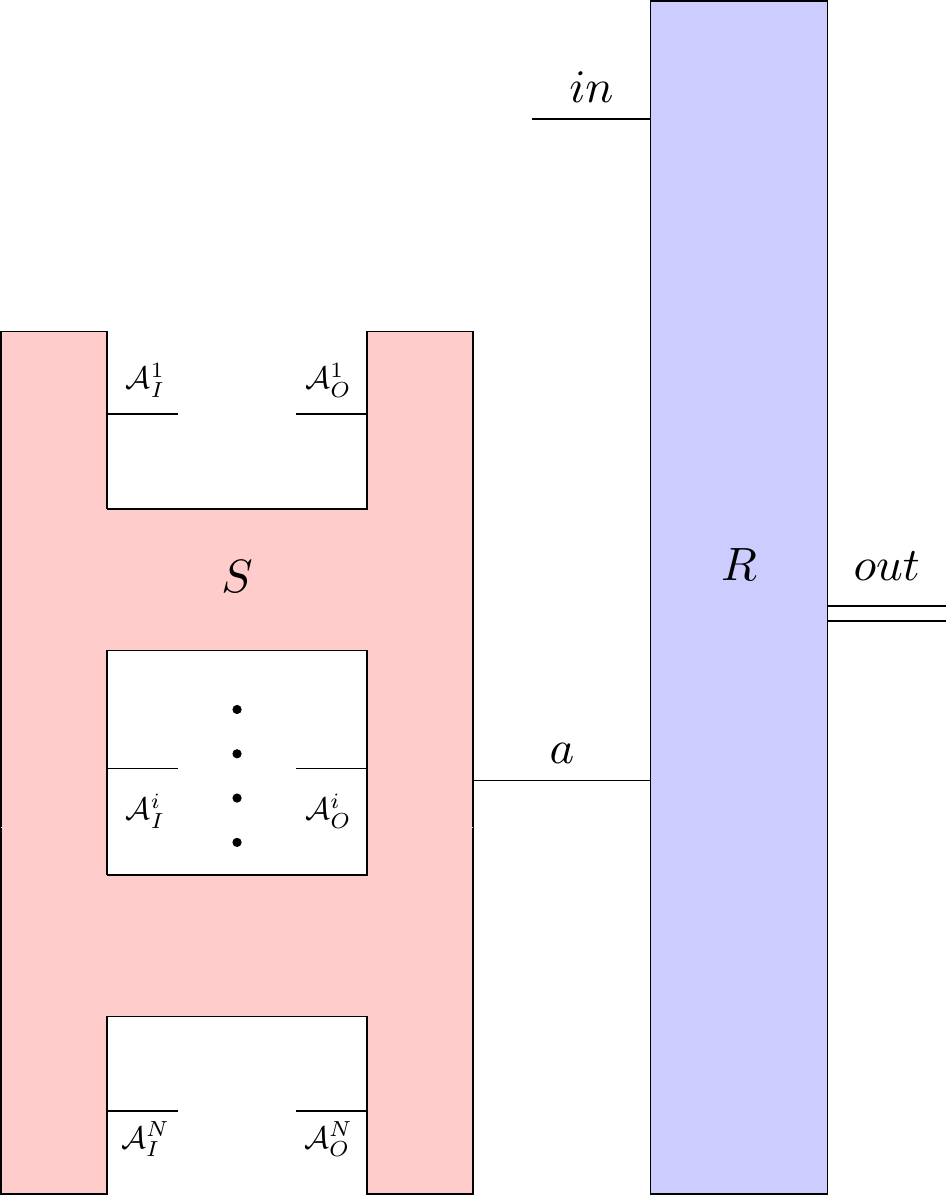}
	\caption{A schematic representation of the setup for  $N \rightarrow 1$ 
	learning scheme of von Neumann
		measurements~$\PP_U$ with the usage of indefinite causal order 
		structures.}\label{fig:schema}
\end{figure} 

As a measure of quality of approximating a von Neumann measurement $\PP_U=\{ 
P_{U,i}\}_i$ with a POVM $\QQ_U = \{ Q_{U,i}
\}_i$ we choose the fidelity function~\cite{raginsky2001fidelity}, which is 
defined as follows
\begin{equation}\label{fidelity}
	\mathcal{F}(\PP_U, \QQ_U) \coloneqq \frac{1}{d} \sum_i
	\tr(P_{U, i} Q_{U, i}),
\end{equation}
where $d$ is the dimension of the measured system. 
Note that in the case when $\PP_{U}$ is a von Neumann measurement we obtain the 
value of fidelity 
function $\mathcal{F}$ belongs to the interval $[0,1]$ and equals to one if and 
only if $P_{U,i} = Q_{U,i}$ 
for all $i$. As there is no prior information about $\PP_U$
provided, we assume that $U$ is sampled from a distribution pertaining to the 
Haar measure.
Therefore, considering a von Neumann measurement $\PP_U$ and its approximation 
$\QQ_U$ we introduce
the \emph{average fidelity function}~\cite{bisio2016quantum} with respect to 
Haar measure as
\begin{equation}\label{average-fidelity}
	\mathcal{F}^{\text{	avg}}_{d} \coloneqq \int_U dU \mathcal{F}(\PP_U, \QQ_U).
\end{equation}

Our main goal is to maximize $\mathcal{F}_\text{avg}$ over all possible 
learning schemes $\LL$.  We introduce the 
notation of the
maximum value of the average fidelity function
\begin{equation}\label{eq:fidelity}
	F_d(N) \coloneqq \max_\LL \mathcal{F}_\text{avg}.
\end{equation}

\section{$N \rightarrow 1 $ learning scheme of von Neumann measurements 
}\label{sec:cases}

This section presents the various results for  $N \rightarrow 1 $ learning scheme 
of von Neumann measurements. 
The solution for one copy of von Neumann measurements was provided in 
\cite{bisio2016quantum}. The Authors then have proved that $F_d(1) = 
\frac{d+1}{d^2}$. 
But what if we have access to more copies of von Neumann measurements? 

\subsection{ $2\rightarrow 1 $ learning scheme}\label{sec:2-to-1}
Let us consider a  learning scheme  $\mathcal{L}$ in which we learn a von 
Neumann measurement $\PP_U$ by using two copies of it.  
Then, its Choi  operator $L \in \mathrm{L}(\A_I \otimes \A_O \otimes \B_I \otimes 
\B_O \otimes \XX_{in} \otimes \XX_{out})$ satisfies the condition
\begin{equation}
\tr_{\XX_{out}} L = \Id_{{\XX_{in}}} \otimes W, 
\end{equation}
where 
$W$ is a bipartite process matrix $(N=2)$. The following proposition says 
that no matter which approach we use, either parallel or indefinite causal 
learning strategies,  the maximum value of the average fidelity for $2 
\rightarrow 1$ learning scheme of von Neumann measurements is the same. The 
solution for the parallel case was provided in \cite{bisio2016quantum}.

\begin{proposition}\label{causal-learning-scheme-2-1}
	The usage of indefinite causal order structure does not improve the maximum 
	value of the average fidelity 
	function in    
	$2 \rightarrow 1$ learning scheme of von Neumann measurements $\PP_{U}$ of 
	dimension $d$.
\end{proposition}

\begin{proof}
	Due to the fact that the quantum network $\mathcal{L}$ has classical 
	labels on spaces $\A_O, \B_O$ and $\XX_{out}$ then its Choi matrix $L \in  \mathrm{Herm}(\A_I \otimes \A_O \otimes \B_I \otimes \B_O \otimes \XX_{in} \otimes \XX_{out})$ has the following form
$
	L = \sum_i \proj{i}_{\XX_{out}} \otimes L_i, 
$
	where $
	L_i = \sum_{j,k} \proj{j}_{\B_O} \otimes \proj{k}_{\A_O} \otimes L_{ijk}. 
$
	Simultaneously,  we know that
$
	\sum_i 	L_i = \sum_{i, j,k} \proj{j}_{\B_O} \otimes \proj{k}_{\A_O} \otimes L_{ijk} = \Id_{\XX_{in}} \otimes W.
$
	Without loss of generality we can assume that $L_{ijk} $ satisfies the commutative relation (see Lemma 9.3 in \cite{bisio2016quantum})
	\begin{equation}
	\left[L_{ijk}, U_\mathcal{A} \otimes U_{\mathcal{B}} \otimes U^\dagger \right] = 0,
	\end{equation} where $U_{\mathcal{A}} \in \mathrm{L}(\A_I \otimes \A_O)$, $U_{\mathcal{B}} \in \mathrm{L}(\B_I \otimes \B_O)$ and $U^\dagger \in \mathrm{L}\left({\XX_{in} \otimes \XX_{out}}\right)$. 
	So, we have
$
	\sum_i L_{ijk} = \Id_{\XX_{in}}\otimes  \bra{j}_{\B_O} \bra{k}_{\A_O} W \ket{j}_{\B_O} \ket{k}_{\A_O}.
$ for all $j,k$.  
	From relabeling symmetry property  (see Lemma 9.4 in \cite{bisio2016quantum}) given by $   L_{ijk} = L_{\sigma(i) \sigma(j) \sigma(k)} $, we have 
	\begin{equation}
	\begin{split}
		\sum_i L_{ijk} & = \sum_i L_{\sigma(i) \sigma(j) \sigma(k)}   \\& = \Id_{\XX_{in}} \otimes \bra{\sigma(j)} \bra{\sigma(k)} W \ket{\sigma(j)} \ket{\sigma(k)},
	\end{split}
	\end{equation}
	for any permutation $\sigma$. 
	Therefore, we have
	\begin{equation}\label{eq:w}
	\begin{split}
	W &  = \sum_{j,k} \proj{j}_{\B_O} \otimes \proj{k}_{\A_O} \otimes \frac{1}{d} \sum_i \tr_{\XX_{in}} L_{ijk} \\& =
	\sum_{j,k}  \proj{j}_{\B_O} \otimes \proj{k}_{\A_O} \otimes  W_{jk}. 
	\end{split}
	\end{equation} 
	Hence,
$
	\forall j,k,\sigma \,\,\ W_{jk} = W_{\sigma(j) \sigma(k)}.
$
	It implies that
	$
	W_{11} = \ldots = W_{dd} $ and $ 
	W_{12} = W_{ab} $ for all $  a \neq b$. These properties together with Eq.~\eqref{eq:w} imply that
$
	W = \Id_{\B_O} \otimes \Id_{\A_O} \otimes P + J(\Delta) \otimes \left( Q - P \right), 
$
	where $ P = W_{12}$ and $Q = W_{11}$ and $J(\Delta)$ denotes the Choi matrix of of the completely
	dephasing channel $\Delta$. 
	From the definition of the process matrix (more precisely from the condition $W 
	+ \prescript{}{\A_O\B_O}{W} 
	= \prescript{}{\A_O}{W}  + \prescript{}{\B_O}{W}$) we obtain that $P = Q$, which completes the proof. 
\end{proof}

\subsection{Semidefinite program  for calculating the maximum value of the average 
fidelity }\label{sec:sdp}

In the general approach, to compute the maximum value of the average fidelity 
$F_d(N)$ we use the
semidefinite programming (SDP). We will present the original primal problem 
(Program~\ref{sdp-1}) for computing $F_d(N)$ for $N \rightarrow 1 $  learning 
scheme 
of von Neumann measurement $\PP_{U}$ of dimension $d$. Next, we will describe a 
simplified version of the primal problem presented in Program~\ref{sdp-2}.

To optimize this problem, we used the \texttt{Julia}
programming language along with quantum package \texttt{QuantumInformation.jl} \cite{Gawron2018} and
SDP optimization via SCS solver \cite{ocpb:16, scs} with absolute convergence tolerance $10^{-5}$. The
code is available on GitHub \cite{code22}. 

\begin{table}[htp!]
	\begin{center}
		\centering\underline{Original problem}
		\begin{equation*}
		\begin{split}
		\text{maximize:}\quad &
		\int_U dU \frac{1}{d} \sum_{i=1}^d
		\tr\left[L_i^\top \left(P_{U, i} \otimes P_U^{\otimes N}\right)\right]
		\\[2mm]
		\text{subject to:}\quad & 
		L \in \mathrm{Pos}\left(\bigotimes_{i=1}^N \A_{I,O}^i \otimes 
		\XX_{in} \otimes \XX_{out}\right), \\ &
		L = \sum_{i=1}^d \proj{i}_{\XX_{out}} \otimes L_i, \\ &
		\prescript{}{\XX_{out}}{L} = \prescript{}{\XX_{out}, \XX_{in}}{L}, 
		\\ &
		\prescript{}{\XX_{out}}{L} = \prescript{}{[\Id - \Pi_i(\Id - \A_O^i + 
			\A_{I,O}^i) + \Pi_i 
			\A_{I,O}^i]}{\left[\prescript{}{\XX_{out}}{L}\right]},\\ &
		\tr(L) = d^{N+1}.
		\end{split}
		\end{equation*}
		\caption{Semidefinite program for maximizing the value of the  average 
		fidelity function $F$ for $N \rightarrow 1 $  learning scheme of von 
		Neumann measurement $\PP_{U}$ of dimension~$d$. }
		\label{sdp-1}
	\end{center}
\end{table} 

Here, we present
a simplified description of the primal problem associated with Program \ref{sdp-1}. 
Let $\YY = \sum_{j\in\XX} 
\left(\sum_{i=1}^d \sum_{j \not\in 
	\XX} \tr_{\XX_{in}, \Pi_{\not\XX}\A_I^i}L_{i,j}\right) 
\otimes 
\proj{j}_{\XX} $ such that $\XX \neq \emptyset$.  Then, we have:

\begin{table}[htp!]
	\begin{center}
		\centering\underline{Simplified problem}
		\begin{equation*}
		\begin{split}
		\text{maximize:}\quad &
		\frac{1}{d} \sum_{i=1}^d \sum_{j_1,\ldots,j_N=1}^d
		\tr\left[L_{i,j} \proj{i} \otimes \proj{j}\right]
		\\[2mm]
		\text{subject to:}\quad & 
		L_{i,j} \in 
		\mathrm{Pos}\left(\XX_{in} \otimes \A_{I} 
		\right), \\ &
		\sum_i L_{i,j} = \prescript{}{\XX_{in}}{\left[\sum_i L_{i,j} \right] }
		\quad \forall_j,
		\\ &
		[L_{i,j}, \bar U \otimes U^{\otimes N}] = 0, \\& 
		\prescript{}{[\Pi_{\XX}(\Id - \A_O^i)]}{\left[ \YY  \right]} = 0, \\&
		\sum_{i,j} \tr(L_{i,j}) = d^{N+1}.\\ &
		\end{split}
		\end{equation*}
		\caption{A simplified description of the SDP Program \ref{sdp-1}.}
		\label{sdp-2}
	\end{center}
\end{table}   

\begin{remark}
	We would like to point out that the commutation relation $[L_{i,j}, \bar U 
	\otimes U^{\otimes N}] = 0$ can be equivalently exchanged with $L_{i,j} = 
	\bigoplus_{\mu \in \mathrm{irrepS(U 
			\otimes U^{\otimes N})}} \Id_{d_\mu} \otimes Q_{i,j, \mu}$, where the 
	summation goes over the irreducible representation of 
	$L_{i,j}$~\cite{bisio2016quantum}.
	
\end{remark}

\subsection{ $N\rightarrow 1 $ learning scheme}\label{sec:n-to-1}
In this section, we analyze the asymptotic behavior of
$F_d(N)$ for $ N \rightarrow \infty$.  Our main result can be summarized as
the following theorem.

\begin{theorem}
	Let $F_d(N)$ be the maximum value of the average fidelity function, defined in 
	Eq.~\eqref{eq:fidelity} for $N \rightarrow 1 $  learning scheme of von Neumann 
	measurements. Then, for
	arbitrary but fixed dimension $d$ we obtain 
	\begin{equation}
	F_d(N) = 1 - \Theta \left(\frac{1}{N^2} \right). 
	\end{equation}
\end{theorem}
\begin{proof}
We will follow the approach from~\cite[Lemma 3]{lewandowska2022storage}. Based on 
the results from~\cite{lewandowska2022storage} it is enough to show $F_d(N) \le 1 - 
\Theta \left(\frac{1}{N^2} \right).$ 

A learning network $L$ can be described as a concatenation of a storage $S$ 
and a retrieval $R$, that is $L = R * S$. What is more, we can assume that 
storage is given as a purification, so the Choi-Jamio{\l}kowski isomorphism 
of $S$ is pure, $S = \projv{X} \in \A_{I, O} \otimes \XX_a$ and $X \ge 
0$~\cite[Lemma 3]{lewandowska2022storage}, 
see 
Fig.~\ref{fig:schema}. Then, $W = \tr_{\XX_{a}}(S) = X^2.$ From the SDP program, 
we have $[W, \Id_{\A_O} \otimes U^{\otimes N}] =0$ for each unitary matrix 
$U$. Therefore, $[X, \Id_{\A_O} \otimes U^{\otimes N}] =0$ and the memory 
state that keeps the information of $P_U$ has the 
form
\begin{equation}
\begin{split}
&\tr_{\A_{I, O}} \left(S \left(\Id_{{\XX_{a}}} \otimes \left(\Id_{\A_O} \otimes 
U^{\otimes N} J(\Delta) \Id_{\A_O} \otimes {U^\dagger}^{\otimes 
N}\right)\right)\right) = \\
& \left(\Id_{\A_O} \otimes \overline{U}^{\otimes N}\right) \rho \left(\Id_{\A_O} 
\otimes 
{U^\top}^{\otimes N}\right),
\end{split}
\end{equation}
where $\rho$ is some state. That means, we can upper bound the value of the 
fidelity within the new scheme, where we are given in parallel 
$N$ copies of unitary channel $\Phi_{\overline{U}}$ and we try to learn 
$\PP_{U}$. According, to \cite[Lemma 7]{lewandowska2022storage} we get 
$F_d(N) \le 1 - \Theta \left(\frac{1}{N^2} \right).$
\end{proof}

\section{Numerical results for qubit von Neumann measurements}\label{sec:numerical}

Although the maximum value of the average fidelity function behaves asymptotically the same using quantum combs or indefinite causal order, we will show here a numerical advantage of SAR for qubit von Neumann measurements  ($d=2$). To show that, we compare
the results for learning scheme for $N \ge 3$ with the parallel and 
adaptive learning schemes introduced in \cite{lewandowska2022storage}. 

In the qubit case, we make two simplifications of SDP Program~\ref{sdp-2}. 
First, the relation $\forall_U \,\, [L_{i,j}, \bar U \otimes U^{\otimes N}] 
= 0$ is equivalent with $\forall_U \,\, [L_{i,j}, U^{\otimes N+1}] = 0$. 
Second, as $L = \sum_{i,j} \proj{i}_{\XX_{out}} \otimes \proj{j}_{\A_O} 
\otimes L_{i,j}$, then $W = \sum_{j} \proj{j}_{\A_O} \otimes \frac12 
\tr_{\XX_{in}}\left( \sum_i  L_{i,j} \right)$ is a $N$-partite 
block-diagonal process matrix. Describing $W$ in the Pauli 
basis~\cite{bavaresco2021strict}
$\left\{\Id_2, \sigma_X, \sigma_Y, \sigma_Z \right\}^{\otimes 2N}$, we get 
that $W$ belongs to the subspace spanned by $\left\{\Id_2, \sigma_Z 
\right\}^{\otimes 
N} \bigotimes \left\{\Id_2, \sigma_X, \sigma_Y, \sigma_Z \right\}^{\otimes 
N}$ and within this subspace it is orthogonal to $$\mathrm{R} = 
\{\sigma_Z^{k_1} \otimes 
\cdots \otimes\sigma_Z^{k_N} \otimes M_1^{k_1} \otimes \cdots \otimes 
M_N^{k_N}:$$ $$ 0 \neq k = (k_1,\ldots,k_N) \in \{0,1\}^{ N}, M_n^k \in \{\Id_2, 
\sigma_X, 
\sigma_Y, \sigma_Z\} \}.$$
Verifying if $W$ is orthogonal to $\mathrm{R}$ is easier than verifying  
$\prescript{}{[\Pi_{\XX}(\Id - \A_O^i)]}{\left[ \YY  \right]} = 0$. It can 
be done simply by checking if
$$\sum_{j} (-1)^{k \cdot j} 
\tr_{\XX_{in}, \Pi_{l: k_l = 0}\A_{I}^l} \left( \sum_i  L_{i,j} \right) = 0,
$$
for each $0 \neq k \in \{0,1\}^{ N}$.

Below we present numerical results for $d=2$ and $N=1,\ldots,5$ that compares 
different storing strategies ($\mathcal{L}_{\text{Causal}}$ - indefinite causal 
order learning strategies, $\mathcal{L}_{\text{Adaptive}}$ - adaptive learning 
strategies, $\mathcal{L}_{\text{Parallel}}$ - parallel learning strategies) and 
show the advantage of using indefinite causal 
order strategies.

\begin{equation*}
\begin{array}{l|c|c|c|c|c}
N&1&2&3&4&5\\\hline
\mathcal{F}_2^{\text{avg}}(\mathcal{L}_{\text{Causal}})&
0.7500 &
0.8114&
0.8698 &
0.8981 & 0.9204 
\\\hline
\mathcal{F}_2^{\text{avg}}(\mathcal{L}_{\text{Adaptive}})&
0.7499&
0.8114&
0.8684&
0.8968&
0.9189\\\hline
\mathcal{F}_2^{\text{avg}}(\mathcal{L}_{\text{Parallel}})&
0.7499&
0.8114&
0.8676&
0.8955&
0.9187\\\hline
\end{array}
\end{equation*} 

The above results are presented also in Figure~\ref{fig:learning-causal-advantage}.  

\begin{figure}[htp!]
	\centering	\includegraphics[scale=0.55]{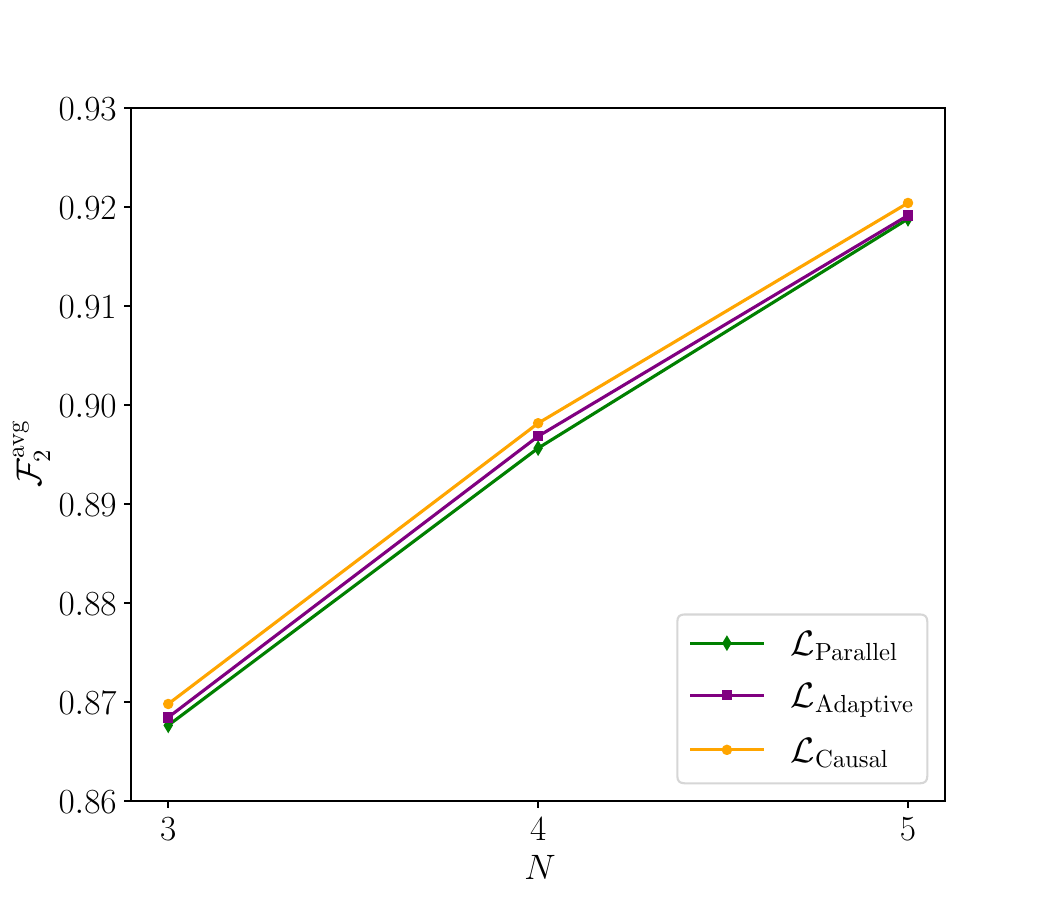}
	\caption{The average fidelity function $\mathcal{F}_2^{\text{avg}}$ 
		 for $N \rightarrow 1$ learning scheme of qubit von Neumann measurements, 
		 where $N = 3,4,5$ -- optimal indefinite causal order learning 
		strategy $\mathcal{L}_{\text{Causal}}$ (orange circles); optimal adaptive 
		learning strategy $\mathcal{L}_{\text{Adaptive}}$ (purple squares); 
		optimal parallel learning scheme $\mathcal{L}_{\text{Parallel}}$ (green  
		triangles). 
	}\label{fig:learning-causal-advantage}
\end{figure}

\section{Conclusions and discussion}\label{sec:conclusion}

In this work,  we studied the problem of learning an unknown von Neumann 
measurement from a finite number of copies $N$ using indefinite causal order 
structures. To do so, we have introduced a notion of $N$-partite process matrices. 
Our main goal was to compute the maximum value of the average fidelity function 
$F_d(N)$ of the approximation, having access to $N$ 
copies of the given measurement. We have proved that $F_d(N) = 1 - \Theta \left( 
\frac{1}{N^2}\right)$  for arbitrary but fixed dimension $d$. 
Next, we have considered various learning schemes for different numbers of 
accessible copies of von Neumann measurements. For $N=2$, we proved that using an 
indefinite causal structures do not improve the average fidelity function $F_d(2)$. 
Next, however, we show numerically advantage of using indefinite causal order 
structures for $d=2$ and $N \ge 3$. For this purpose, we have stated a SDP program 
and provided its simplified version.

Our results give additional confirmation of potential benefits of using indefinite 
causal order structures in quantum information theory. Previously, applications 
were observed for example in quantum channel 
discrimination~\cite{bavaresco2021strict}, quantum 
communication~\cite{ebler2018enhanced} or unitary channels 
transformation~\cite{quintino2019probabilistic, quintino2019reversing}. Here, we 
showed the usage of the theory of indefinite causal order in the problem of 
learning of quantum channels, in particular by using storage and retrieval scheme.

\section*{Acknowledgments}

PL is supported by the Ministry of Education, Youth and Sports of the Czech Republic through the e-INFRA CZ (ID:90254),
with the financial support of the European Union under the REFRESH – 
Research Excellence For REgionSustainability and High-tech Industries 
project number CZ.10.03.01/00/22\_003/0000048 via the Operational Programme 
Just Transition.

RK is supported by the National Science Centre, Poland, under the contract 
number 2021/03/Y/ST2/00193 within the QuantERA II Programme that has 
received funding from the European Union’s Horizon 2020 research and 
innovation programme under Grant Agreement No 101017733.

\bibliographystyle{ieeetr}
\bibliography{learning.bib}

\onecolumngrid
\newpage
%
%
%

\end{document}